\documentclass[preprint,1p]{elsarticle}

\usepackage[dvipsnames,usenames]{color}
\usepackage{caption}
\usepackage{amssymb,amsmath,amsthm}
\usepackage{mathrsfs}
\usepackage{mathtools}
\usepackage{hyperref} 
\hypersetup{
	colorlinks=true,
	linkcolor=blue,
	filecolor=magenta,      
	urlcolor=blue,
	pdftitle={Overleaf Example},
	pdfpagemode=FullScreen,
	citecolor=blue
}

\usepackage{cleveref} 

\theoremstyle{plain}
\newtheorem{theorem}{Theorem}[section]
\newtheorem{lemma}[theorem]{Lemma}
\newtheorem{prop}[theorem]{Proposition}
\newtheorem{cor}[theorem]{Corollary}

\theoremstyle{remark}
\newtheorem{defn}{Definition}

\newtheorem{rmk}{Remark} 

\newcommand{\LSA}{{\operatorname{LSA}}}

\newcommand{\cf}{\mathrm{CF}}

\newcommand{\popsize}{\eta} 

\newcommand{\Pf}{{\mathrm{I}\!\mathrm{P}}}
\renewcommand{\P}[1]{\Pf\!\left\{#1\right\}}

\DeclarePairedDelimiter\abs{\lvert}{\rvert}
\makeatletter
\let\oldabs\abs
\def\abs{\@ifstar{\oldabs}{\oldabs*}}
\makeatother

\definecolor{olive}{rgb}{0.42, 0.56, 0.14}
\usepackage{xcolor}

\definecolor{steelblue}{rgb}{0.27, 0.51, 0.71}

\definecolor{changecolor}{RGB}{192,64,0}

\usepackage[normalem]{ulem}
\newcommand{\com}[1]{} 

\begin{document}

\begin{frontmatter}
\title{Beyond level-1: Identifiability of a class of galled tree-child networks}

\author[1]{Elizabeth S. Allman}\ead{esallman@alaska.edu}

\author[2]{C\'ecile An\'e}\ead{cecile.ane@wisc.edu}

\author[3]{Hector Ba\~{n}os}\ead{hector.banos@csusb.edu}

\author[1]{John A. Rhodes\corref{cor1}}\ead{jarhodes2@alaska.edu}

\cortext[cor1]{Corresponding author. All authors contributed equally and are listed alphabetically.}
\affiliation[1]{
organization={Department of Mathematics and Statistics, University of Alaska - Fairbanks},
postcode={99775-6660}, state={AK}, country={USA}}
\affiliation[2]{
 organization={Department of Statistics, University of Wisconsin - Madison},
postcode={53706}, state={WI}, country={USA}}
\affiliation[3]{
organization={Department of Mathematics, California State University, San Bernardino},
postcode={92407}, state={CA}, country={USA}}

\begin{abstract}
Inference of phylogenetic networks is of increasing interest in the genomic era.
However, the extent to which phylogenetic networks are identifiable from
various types of data remains poorly understood, despite its crucial role in  justifying methods.
This work obtains strong identifiability results for large sub-classes of
galled tree-child semidirected networks.
Some of the conditions our proofs require, 
such as the identifiability of a network's tree of blobs or the
circular order of 4 taxa around a cycle in a level-1 network,
are already known to hold for many data types.
We show that all these conditions hold for quartet concordance factor data
under various gene tree models, yielding the strongest results
from 2 or more samples per taxon.
Although the network classes we consider have topological restrictions, 
they include non-planar networks of any level and are substantially more general than level-1 networks
--- the only class previously known to
enjoy identifiability from many data types.
Our work establishes a route for proving future identifiability results for tree-child
galled networks from data types other than quartet concordance factors, by checking that explicit
conditions are met.
\end{abstract}

\end{frontmatter}

\section{Introduction}

The analysis of genomic data sets in recent years has led to the discovery
of numerous instances of hybrid speciation or gene flow,
with phylogenetic networks and admixture graphs
increasingly used to describe evolutionary relationships
\cite[e.g.,][]{Linder2004,Mallet2005,Noor2006,DeRaad22,Lopes23,Yang23,2023Nielsen-admixturebayes,maier2023limits,Ciezarek24}.
Several inference methods used in these analyses
have focused on the simplest class of networks, those of level 1,
in which reticulations are sufficiently isolated from one another that the network
shows only disjoint cycles joined by tree-like edges.
While computational difficulties have been partially responsible for this focus,
a lack of theoretical understanding of the extent to which more
complex network classes are identifiable is also a barrier.

Network identifiability, which is the property that sufficiently large
data sets produced in accord with a model allows for the network's
recovery in principle, is essential for valid statistical inference.
For level-1 networks, identifiability has been proved
(sometimes requiring mild restrictions) under various models and data types
\cite{SolisLemus2016,GrossLong18,Banos2019,Allman2019,GrossEtAl2021,ABR2022logdetNet,2023XuAne_identifiability,2024ABGarrotelopesR}.

In this work, we extend our understanding of network identifiability substantially,
to a large class of networks, those whose blobs are galled and tree-child.
Informally, galled networks are those for which every reticulation lies in a
cycle with no others, and a tree-child network is one in which every node has at least one
child node that is not a reticulation.
Despite the simple structure of these networks, they can be of arbitrary level, and need not be planar. 
The theoretical results here suggest that these networks should be a good ``next step" class of 
networks on which developers of inference methods might focus.

In establishing our results, we consider quartet concordance factors (CFs)
as input data. These are the proportions of gene trees displaying
the various unrooted 4-taxon tree relationships for each subset of 4 taxa.
Such CFs have formed the basis of a number of inference methods under the
multispecies coalescent model, notably
ASTRAL \cite{Zhang2018} for the inference of species trees, and
SNaQ \cite{SolisLemus2016} and NANUQ \cite{ABR2019,ABRW2025} 
for the inference of level-1 networks.
Quartet information also underlies PhyNEST \cite{Kong2024}, although
with site pattern frequencies as input.
Quartet CFs are attractive for inference for several reasons, including providing a computational speedup over methods that use full gene trees.
In addition, since quartet CFs capture only topological gene tree information,
they offer robustness to variability of
substitution rates across genes and lineages, to departures from
a molecular clock, and to edge length estimation error in gene trees.

\begin{figure}
\centering
\includegraphics[scale=.9]{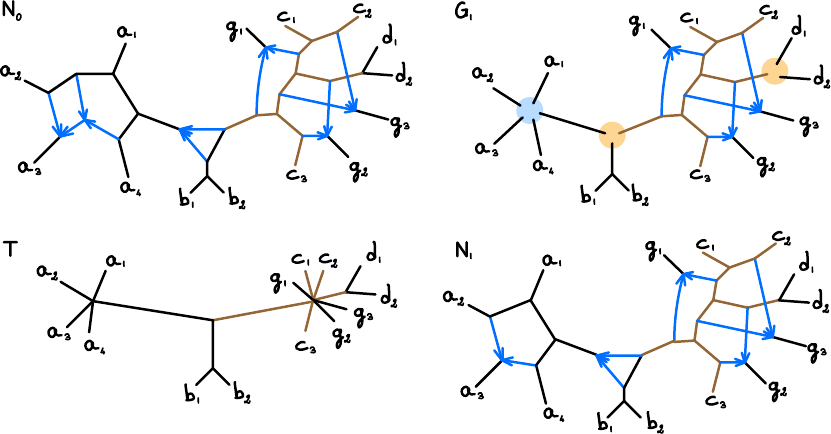}
\caption{Example of a network $N_0$ (top, left) and some of its
features identifiable from quartet concordance factors.
$T$ (bottom, left) is $N_0$'s tree of blobs.  Since the largest (rightmost) 
blob of $N_0$ is $\frak C_5$ (see \Cref{def:classCk}),
that blob's full topology is identifiable using one sample per taxon.
This blob appears in $G_1$ (top, right), which shows the features of $N_0$
proved to be identifiable in this article, 
with large circles representing unresolved blobs of uncertain topology.
As discussed in the main text, 
the central 3-blob can be detected as non-trivial and its hybrid node can sometimes be identified.
In contrast, the left-most blob is not identifiable: $N_1$ (bottom, right)
cannot be distinguished from $N_0$ using quartet concordance factors.
}\label{fig:intro}
\end{figure}
 
\Cref{fig:intro} provides an informal example of our results
which, together with previously established identifiability theorems
\cite{ABMR2023, 2024ABGarrotelopesR,2025Rhodes-circular},
illustrate that many features of complicated networks are in fact identifiable from CFs.
Under commonly-used models of gene tree generation
on the species network $N_0$, with one sample per taxon, the tree of blobs $T$,
in which each blob (defined in  \Cref{sec:skeletons-blobs})  
is contracted to a node is first identified. We then analyze each blob individually and are able to identify the full structure
of the large (rightmost) blob of $N_0$, since it is galled, tree-child, and has no small cycles.
The other blobs of $N_0$, left unresolved and depicted as colored circles
in $G_1$, are either not fully identifiable or require additional assumptions to be so.
For the two 3-blobs in $T$, shown with orange circles in $G_1$,
by taking two samples per taxon, we can identify that the rightmost 3-blob is trivial
and the central 3-blob is not.
If  a non-trivial 3-blob is assumed to be level-1, as in this case,
 its hybrid node is identifiable only for certain edge lengths.  
For the leftmost level-2 blob of $N_0$, since it is outer-labelled planar,
the circular order of taxa  around it can be identified as well as that
taxon $a_3$ descends from a hybrid node. 
However, that blob's internal structure cannot be identified with a single sample per taxon.
Indeed, edge parameters on $N_1$ can be set
(as in \cite[Section 7.2]{2025Rhodes-circular}) such that $N_1$ is
undistinguishable from $N_0$, both generating the same quartet concordance factors.

\smallskip
  
After introducing terminology and requisite background results
(\Cref{sec:networks}) and introducing our new class of networks (\Cref{sec:galledTreeChild}),
our arguments are structured into two parts.
Following the statements of a few general assumptions on a network class,
we give combinatorial arguments for network identifiability under these
(or subsets of these) assumptions in \Cref{sec:id}.
Then, in \Cref{sec:moddat}, we show that for quartet CF data 
these basic assumptions hold under various models,
leading to our main practical result in \Cref{thm:main}. 
A discussion in \Cref{sec:disc} concludes this work.

\section{Phylogenetic networks and blobs}\label{sec:networks}

We use standard  terminology for phylogenetic networks, as in \cite{Steel2016,SolisLemus2016,Banos2019,2024Ane-anomalies}, 
recalling it briefly in the next two subsections. We then define more specialized terms central to the work here, and establish some 
basic properties in the remaining subsections.

\subsection{Rooted networks}
A \emph{rooted topological phylogenetic network} $N^+$ on a set of
taxa $X$ is a finite connected rooted directed acyclic graph 
with vertices comprising a \emph{root},
\emph{hybrid nodes}, \emph{internal tree nodes} and \emph{leaves},
and edges which are either \emph{hybrid} or \emph{tree} edges.
The root has in-degree 0. Leaves are of in-degree 1 and out-degree 0 
and bijectively labeled by elements of $X$.
Hybrid nodes have in-degree at least 2 and out-degree at least 1.
The internal tree nodes make up the remaining nodes.
By \emph{tree nodes}, we mean all non-hybrid nodes.
If a non-root node in $N^+$ has degree 3, then this node is \emph{binary}.
A hybrid node is \emph{bicombining} if its in-degree is 2.
An edge is \emph{hybrid} or \emph{tree} in accord with its child node.
Any hybrid edge shares its child node with at least one other \emph{partner} hybrid edge.
An edge incident to a leaf is called \emph{pendant}.

A network is \emph{binary} if its root (if it has one) has
degree 2 and all other non-leaf nodes are binary. 
A network $N^+$ is  \emph{metric} if  
each edge $e$ is assigned a pair of parameters
$(\ell(e),\gamma(e))$, where $\ell(e)\ge 0$ is an \emph{edge length}
with $\ell(e)> 0$ for tree edges, and
$\gamma(e)\in(0,1]$ is a \emph{hybridization} or \emph{inheritance parameter},
with the sum over partner edges equal to 1.
Note that $\gamma(e)=1$ if $e$ is a tree edge,
and we  require that $\gamma(e) < 1$ if $e$ is a hybrid edge.

We say a node or edge $s$ is \emph{above} or \emph{ancestral}  to another node or edge $s'$ in $N^+$
(and $s'$ is \emph{below} $s$ or a \emph{descendant} of $s$)
if there is a (possibly empty) directed path from $s$ to $s'$.
An \emph{up-down path} or \emph{trek} in $N^+$ is an undirected
path joining nodes $(u_1,u_2,\dots,u_n)$
such that for some $i\in[n]$, $u_i\dots u_2u_1$ and $u_i\dots u_{n-1}u_n$
are directed paths in $N^+$.

Let $N^+$ be a network on $X$ and let $Y\subseteq X$.  
The \emph{least stable ancestor} of $Y$ on $N^+$, denoted $\LSA(Y)$,
is the lowest node through which all directed paths from the root
to any taxon in $Y$ must pass. The \textit{LSA network of $N^+$},
is the network obtained from $N^+$ by deleting all edges and nodes
strictly ancestral to $v=\LSA(X)$, and rerooting at $v$.
We say  $N^+$ is an \emph{LSA network} if $r=\LSA(X)$. 

Throughout this work we assume all rooted networks are LSA networks.

\subsection{Semidirected networks}
\label{sec:semidirectednetworks}

The \emph{semidirected} phylogenetic network $N^-$ of a (LSA) network $N^+$
is the graph obtained by undirecting all tree edges and suppressing its root if it has degree 2.
The network  $N^+$ is an example of a \emph{rooted partner of $N^-$} \cite{2023LinzWicke}.
A semidirected phylogenetic network may have more than one
rooted partner, since many rootings of $N^-$ might be consistent
with hybrid edge directions in $N^-$.
 
An \emph{up-down path} or \emph{trek} in $N^-$ is a path in the fully undirected
graph without any pair of consecutive hybrid edges $p_1 h, h p_2$ directed as $(p_1,h)$, $(p_2,h)$.
As shown in \cite{2023XuAne_identifiability}, up-down paths in $N^+$ are in 
bijection with up-down paths in $N^-$.
A \emph{semidirected path} in $N^-$ is a path
joining nodes $(u_1,u_2,\dots,u_n)$ such that each edge
is either undirected, or directed as $(u_i,u_{i+1})$.
An \emph{up-down cycle} or  \emph{trek cycle} is 
a non-empty up-down path starting and ending at the same
node, which is necessarily hybrid.

Removing all  hybrid edges from $N^-$ results in connected components (called skeleton trees below),
with exactly one of these, the  \emph{root component},
containing possible root locations \cite{2024MaxfieldXuAne}.
The set of edges where $N^-$ might be rooted consists of 
those edges in the root component and the hybrid edges
incident to the root component.
Note that, by definition, any edge not in the root component
has the same direction in all rooted partners $N^+$.

We say that a node or edge $s$ is  \emph{above} or \emph{ancestral to} another node or edge $s'$ in $N^-$
(and $s'$ is \emph{below} or a \emph{descendant of} $s$)
if $s$ is above $s'$ in every rooted partner of $N^-$.
In particular, tree edges in non-root components of $N^-$ are all below
at least one hybrid node.

For  any rooted or semidirected graph  $G$, 
the \emph{reduced graph} of $G$
is obtained from $G$ by suppressing all degree-2 nodes (other than the root). In this work the reduced version of a graph $G$ may be denoted by
$\overline G$ or simply specified in words.
For a rooted or semidirected network $N$, the \emph{induced network} $N_Y$ of a network $N$
is the network obtained from $N$ by retaining only the up-down paths between all pairs of taxa in $Y$.  
In \cite{Banos2019} it was shown that the operations of inducing and semidirecting
commute, that is, as reduced graphs ${(N^+_Y)}^-=(N^-)_Y$.

\medskip

For the remainder of this work, we focus on semidirected phylogenetic networks $N^-$, 
denoted for simplicity by $N$. 

\subsection{Skeletons and Blobs}
\label{sec:skeletons-blobs}

Our arguments will require a number of special subgraphs of phylogenetic networks that we now define.
	
\begin{defn}
The \emph{unreduced skeleton forest} of  $N$ is the graph obtained by removing
all hybrid edges.
The \emph{unreduced skeleton trees} of $N$ are the connected components
of its unreduced skeleton forest. The \emph{unreduced root skeleton tree} is the root component.
The \emph{skeleton forest, skeleton trees} and \emph{root skeleton tree} of $N$
are the corresponding reduced graphs.
Finally, a skeleton tree is \emph{trivial} if it only contains a single node,
or a single edge. 
\end{defn}

Note that, by definition, each edge in a skeleton tree $T$ is composed of
one or more tree edges in $N$.

Skeleton trees of a phylogenetic network are not necessarily phylogenetic trees,
as they may have unlabeled leaves. See, for example, the root skeleton tree of the network $N'$ in \Cref{fig:example01}.
Related concepts like \emph{tree-node forest} and  \emph{tree-node components}
were introduced in \cite{2017Gunawan}, but differ from the definition here
in that hybrid \emph{nodes} are removed, and hence the child edges of those nodes are also deleted.

\begin{lemma}\label{lem:skeletonnum}
If $N$ is a network with $k$ hybrid nodes, then the skeleton forest of $N$ has $k+1$ skeleton trees.
\end{lemma}
\begin{proof} 
 Letting $N^+$ be  a rooted partner of $N$ with root $\rho$, pick a lowest hybrid node $v$ in $N^+$
 (and $N$).
 By deleting all $v$'s partner hybrid edges we obtain a graph with 2 connected components
 since all nodes below $v$ remain connected to $v$, and all other nodes are connected to $\rho$ by
 some path.   The first component has no hybrid nodes and is thus a tree, 
 while the component containing $\rho$ has $k-1$ hybrid nodes.   Repeating this process
 on the component containing $\rho$ until no hybrid nodes remain,
 gives $k+1$ components which are the unreduced skeleton trees.
 \end{proof}
 
\medskip

Let $G$ be an arbitrary graph. 
A \emph{blob} of $G$ is a maximal connected subgraph with no cut edges
(that is, a 2-edge-connected component).
A \emph{trivial blob} is one consisting of a single node.
Note that a non-trivial blob is a biconnected component (or block)
if the network is binary, but otherwise may contain one or more blocks
\cite{2023XuAne_identifiability}. 

A node in a blob is a \emph{boundary node} if
it is incident to one or more cut edges.
A blob incident to exactly $m$ cut edges is an \emph{$m$-blob}.
If a network is binary, a non-trivial $m$-blob has exactly $m$
boundary nodes.
In general, an $m$-blob has $m$ or fewer boundary nodes,
but any network with an $m$-blob must have at least $m$ taxa.
For examples, see \Cref{fig:example01}.

\begin{figure}[h]
\centering\includegraphics[scale=1]{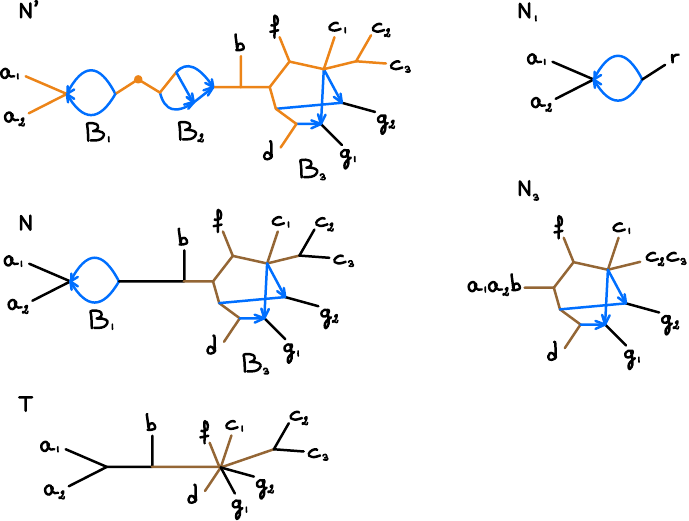}
\caption{ 
The semidirected network $N'$ (top left) is unreduced with a degree-2 node marked with a dot.
Its unreduced skeleton forest, obtained by removing hybrid edges (blue),
has 3 non-trivial unreduced skeleton trees (orange).
$N'$ has 3 non-trivial blobs: $B_1$, $B_2$, and $B_3$.
The blob $B_1$ is a level-1 3-blob with 2 boundary nodes, and
$B_3$ is a 7-blob with 6 boundary nodes.
$N'$ is a 2-blob extension of $N$ (center left).
$T$ (bottom left) is the tree of blobs for both $N$ and $N'$.
The subnetwork of $N'$ induced by $\{a_1, f, c_1, c_2, g_2, g_1, d\}$
is an extended bloblet with $B_3$ its internal blob.
\endgraf
\hspace*{.1in}
In the language of \Cref{sec:galledTreeChild}, $N_1$ and $N_3$ (right) are
bloblets generated by $B_1$ and $B_3$ respectively.
Both are non-binary, but in $N_3$ hybrid nodes are binary.
$N_1$ is a galled, weakly (but not strongly) tree-child bloblet.
$N_3$ is $\mathfrak{C}_k$ for $k=3,4,5$, with its unreduced root skeleton tree
shown in brown.}
\label{fig:example01}
\end{figure}

For a network $N$ the \emph{unreduced tree of blobs} is
the tree obtained from $N$ by contracting each of its blobs to a vertex.
Since blobs are 2-edge connected components, this contraction
is well-defined, and distinct blobs are contracted to distinct vertices
in the unreduced tree of blobs. 
The \emph{tree of blobs} $T(N)$ of $N$ is obtained by reducing its unreduced tree of blobs.
See \Cref{fig:example01} for an example.

\begin{defn}\label{def:extension}
A semidirected phylogenetic network $N'$ is a \emph{2-blob extension}
of a network $N$ if $N$ is obtained from $N'$ by contracting
some of its 2-blobs and suppressing the resulting degree-2 nodes.
\end{defn}

\noindent
If $N'$ is a 2-blob extension of $N$, then $N'$ has the same leaf set and the same
tree of blobs $T$ as $N$,
and any node in $T$ corresponds to the same blob in $N$ and $N'$
(see \Cref{fig:example01}).
\begin{defn}
A semidirected network $N$ is a \emph{bloblet} if,
in addition to the trivial blobs of its leaves, it has
a single additional blob that is an $m$-blob with $m\geq 2$.
A network $N'$ is an \emph{extended bloblet} if it is a 2-blob extension of a 
bloblet $N$ with $m\ge 3$. The \emph{internal blob} of an extended bloblet
is its unique $m$-blob with $m\geq 3$.
\end{defn}

\Cref{fig:example01,fig:example02} show examples of bloblets.
The name ``bloblet" was suggested by the term \emph{sunlet}, 
which has been used in many works for graphs with a single blob that is a cycle
in a binary network (a binary level-1 bloblet).  Note that if $N$ is a bloblet  on $n$ taxa, then its internal blob is an $n$-blob
with $k \le n$ boundary nodes, and its  tree of blobs  is a star tree. 
The same holds when $N'$ is an extended bloblet on $n \ge 3$ taxa: it has an internal $n$-blob
and star tree of blobs.
Finally, if a bloblet $N$'s internal blob is trivial, then $N$ itself is a star tree.

\begin{defn}
A node $v$ in a blob $B$ of a semidirected $N$ is a \emph{lowest node} if,
in every rooted partner $N^+$ of $N$,
$v$ has no proper descendants in $B$.
\end{defn}

Note that every non-trivial blob has at least one lowest node, and 
that all lowest nodes are necessarily hybrid. 
This definition extends that of \cite{2025Rhodes-circular} for
rooted networks, where it was proved that any node
without descendants in a rooted non-trivial blob must be hybrid.
If $v$ is a lowest node of a non-trivial blob $B$ in $N$, then its descendant
nodes and edges are the same across all of $N$'s rooted partners,
and we denote $v$'s descendant leaves by $L(v)$. 

\subsection{Basic blob structure}

\begin{defn}
 Let $N$ be a semidirected network on taxon set $X$, with  $T$ its tree of blobs. 
 Suppose $v$ is a lowest node in a non-trivial blob $B$ of $N$.
 Let $Y = X\smallsetminus L(v)$ and $\overline{T}_Y$ denote the reduced graph of $T_Y$.
 We say that $v$ \emph{links} $B$ if $N_{Y}$'s tree of blobs is more resolved than $\overline{T}_Y$.
 If the tree of blobs  of $N_{Y}$ is equal to $\overline{T}_Y$,
 then $v$ \emph{augments} $B$.
\end{defn}

Since any cut edge in $N$ remains a cut edge in $N_{Y}$
(if present in $N_{Y}$), 
$N_{Y}$'s tree of blobs must have all the edges
from $\overline{T}_Y$, and either be identical to
or more resolved than $\overline{T}_Y$,
and these notions of linking and augmenting lowest nodes are exhaustive.
For examples, note that the unique hybrid node in a 
$k$-taxon sunlet network ($k \ge 5$) is a linking
node, whereas $h_2$ in \Cref{fig:example02} (right) is an augmenting node.
Also, if a blob $B$ has a linking node, there must be at least 4 taxa in $Y$,
implying that $X$ has $n\geq 5$ taxa.
If an $m$-blob $B$ with $m\geq 5$ has a single lowest node,
then this node is not necessarily linking, as shown in \Cref{fig:example02} (right).

Note that these definitions depart from those in \cite{2025Rhodes-circular},
where linking and augmenting lowest nodes are defined in \emph{rooted} networks
with at least two lowest nodes.

\begin{figure}
\centering\includegraphics[scale=1]{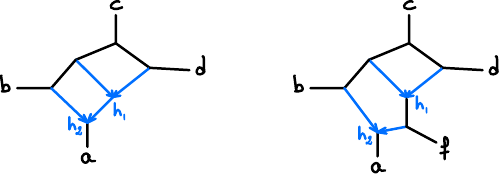}
\caption{Examples of bloblets.
Left: A 4-blob $B$ on a network that is not tree-child and not galled,
with $a$ below $B$'s lowest node $h_2$. 
The root skeleton tree, on the taxa $\{b,c,d\}$, is the star tree,
so $h_2$ is augmenting.
Right: A 5-blob $B$ on a network that is tree-child but not galled.
$B$'s unique lowest node $h_2$ is augmenting.
}\label{fig:example02}
\end{figure}

\begin{lemma}\label{lemma:Linkingnodes}
 Let $N$ be a bloblet on a set $X$, with internal blob $B$.
 Then the tree of blobs of $N_{X\setminus \{x\}}$ is not a star tree
 if, and only if, $x$ is below a linking (lowest) node of $B$, and
$x$ is the only leaf below that node.
\end{lemma}

\begin{proof}
  By hypothesis, the tree of blobs of $N$ is a star tree.
  Let $x\in X$ with $v$ its adjacent internal node, which is a boundary node of $B$,
  and let $T_{x^-}$ denote the tree of blobs of $N_{X\setminus \{x\}}$.
  If $v$ is also adjacent to a second leaf, then $N_{X\setminus \{x\}}$
  contains $v$ and all edges and nodes of $B$, so $T_{x^-}$ is a star.
  Thus, we may hereafter assume that $v$ is adjacent to a single leaf 
  and that $L(v) = \{x\}$.

  Assume first that $v$ is a lowest node of $B$. Then,  by definition, 
  $T_{x^-}$ is a star if $v$ is augmenting and
  $T_{x^-}$ is not a star if $v$ is linking.
  
  Now suppose that $v$ is not a lowest node of $B$.
  Then in every rooted partner $N^+$, $v$ has a descendant edge in $B$.
  If $v$ is hybrid, this follows since $v$ is not lowest.
  If $v$ is a tree node, this follows since $v$ is in $B$:
  If all its descendant edges in $N^+$ were not in $B$, then $v$ would be
  incident to at most 1 edge in $B$,   a contradiction.
  
To show $T_{x^-}$ is the star tree, we will show that (unreduced)
$N_{X\setminus \{x\}}$ cannot have an internal cut edge.
Suppose for the sake of contradiction that it does,
and on a fixed partner $N^+$ rooted at some node $\rho$ of $N$,
pick an edge $e=(s,t)$ that becomes
an internal cut edge in $N^+_{x^-} = N^+_{X\setminus \{x\}}$.
Removing $e$ from $N^+_{x^-}$ leaves two connected components,
$M_s$ containing $s$ and $M_t$ containing $t$. Let $X_s$ be the taxa on $M_s$ and
$X_t$ the taxa on $M_t$, so
$X_s\sqcup X_t=X\setminus\{x\}$.

For use several times in the argument below we claim ($\dagger$): For any taxon $a\in X_s$, no
directed path in $N^+_{x^-}$ ending at $a$ may pass through $e$. To see this, suppose there were
such a path passing through $e$. Then truncating this path gives a path
from $t$ to $a$ in $N^+$ that avoids $e$.
But from $t$ there is also a directed path in $M_t$ to some $b \in X_t$.
After truncating at a lowest common node, these two paths combine to give an up-down path
from $a$ to $b$, which is therefore in $N^+_{x^-}$.
But this shows $M_s$ and $M_t$ are connected in $N^+_{x^-}$  by an up-down path avoiding $e$,
which contradicts that $e$ is a cut edge of $N^+_{x^-}$.

Since $e$ is not a cut edge in $N^+$, by Lemma~10 in \cite{2024Ane-anomalies} there is an up-down cycle $C$
in $N^+$ containing $e$.   Pick such a $C$
with lowest and highest nodes $\ell$ and $h$, respectively.
There is some taxon below $\ell$, so first suppose $x$ is.
Then since $x$ has $v$ as its only parent,  $v$ is also below $\ell$.
But since there is a child edge of $v$ in the blob,
$v$ must also be above another taxon.
Thus whether $x$ is below $\ell$ or not, there is some
other taxon $b$ below $\ell$. In fact $b$ must be in $X_t$,
since if $b$ were in
$X_s$ there would be a directed path from  $h$  through the part of $C$ containing $e$ to $\ell$ and then to $b$, 
contradicting ($\dagger$). Thus $\ell$ is above some $b\in X_t$.

Next consider a directed path avoiding $e$ from the root $\rho$
to $h$ then through part of $C$ to $\ell$ and on to $b$.
For any $a\in X_s$ choose a directed  path from $\rho$ to $a$.
By construction for the first, and  by ($\dagger$) for the second, neither of these
paths pass through $e$. Truncating them at their lowest common node yields an up-down path from $a$ to $b$ 
which is therefore in $N^+_{x^-}$. But since $e$ is not on this
path between $M_s$ and $M_t$ this contradicts that $e$ is a cut edge of $N^+_{x^-}$.
\end{proof}

A final structure lemma we need is the following.

\begin{lemma}\label{lem:3cycin}
Let $N$ be binary bloblet on 3 taxa, with a non-trivial 3-blob.
Then $N$ has a level-1 subnetwork with a single 3-cycle and no 2-cycles.
In other words, $N$ has a 3-sunlet subnetwork.
\end{lemma}

\begin{proof} 
We proceed by induction on the number of hybrid nodes in the 3-blob. 
For the base case, when the 3-blob has a single hybrid, $N$ is a 3-cycle network and there is nothing to prove.

Assuming now that $N$'s 3-blob has more than one hybrid node, then
there exists a lowest hybrid node with a descendant taxon, say $a$.
Let $M=N_{\{b,c\}}$ be the subnetwork composed of all edges on
up-down paths connecting the taxa $b$ and $c$.
Thus $M$ has the form of a chain of 2-blobs (some possibly trivial) joined by cut edges.
The \emph{funnel} of $a$ in $N$ is all edges in up-down paths from $a$ to $M$,
terminating at the funnel's \emph{attachment nodes} on $M$, which in number are at least 2.
If any attachment node $v$ of $a$'s funnel is in a non-trivial 2-blob of $M$,
then we may pick one path in the funnel from $v$ to $a$;
and removing from $N$ all other funnel edges not on this path
gives a subnetwork with 1 fewer hybrid node
and a non-trivial 3-blob, and possibly some 2-blobs. By choosing some semidirected up-down path
 through each 2-blob and deleting all edges in 2-blobs not on these paths, we reduce to a network
for which the inductive hypothesis applies.

Otherwise all attachment nodes are trivial 2-blobs of $M$, at which two cut edges of $M$ join. 
Thus any up-down path in $M$ between
$b$ and $c$ contains all attachment nodes of $a$'s funnel.
Picking two of these attachment nodes, $v,w$ and retaining only a path from $b$ to $c$ and one path each 
from $v,w$ to $a$ in the funnel yields the desired subnetwork.
\end{proof}

\section{Galled and tree-child semidirected networks}\label{sec:galledTreeChild}

A rooted phylogenetic network $N^+$ is  \emph{tree-child}, if every
non-leaf node $v$ has a child that is a tree node  \cite{Steel2016}.
A semidirected network $N$ is \emph{strongly tree-child} (or simply \emph{tree-child})
if all its rooted partners are tree-child, and \emph{weakly tree-child} if at least one of its rooted partners is tree-child
\cite{2024MaxfieldXuAne}.  Examples of semidirected networks that are strongly, weakly, or not tree-child 
are given in \Cref{fig:example1galled}.
A sunlet with 3 or more leaves is strongly tree-child.

A (rooted or semidirected) network $N$ is \emph{galled} if, for every hybrid node $h$
and every pair of partner hybrid edges $e =(v,h)$ and $e'=(v',h)$,
there exists a cycle in $N$ (considering edges as undirected)
that contains $e$ and $e'$ and no other hybrid edges.
Such a cycle is  called a \emph{tree cycle}. 
For example, the network $N_3$ in \Cref{fig:example01} is galled, but
neither of the networks in \Cref{fig:example02} are galled.
The term ``galled network'' should not be confused with ``galled trees,''
a class of networks now commonly referred to as level-1 networks (when binary).  

\smallskip

In the next section, we  prove that galled tree-child bloblets
are identifiable --- provided additional assumptions hold ---
and then extend these results to networks with multiple internal blobs.
In preparation for this, we establish some key properties and introduce
definitions leading to a new subclass of galled tree-child networks.
Some of the properties of galled networks we need have already been developed, for example in
\cite{2007HusonKlopper,HusonRuppScorn,2017Gunawan,2020Gunawan}, 
though we generally give self-contained proofs for the sake of readability.

Let $N$ be a galled network  with  hybrid node $h$.  Then $h$ is a lowest node of the blob containing it
and,  if in addition $h$ is the child of exactly two hybrid edges, 
then $h$ is in a \emph{unique} (tree) cycle.

If $N$ is a galled bloblet, then each hybrid node and its children are the nodes of one skeleton tree,
with all remaining nodes contained in (and connected by) the
unreduced root skeleton tree.   \Cref{fig:example1galled}  illustrates the root skeleton tree may or may not be trivial, and can have 0, 1 or more taxa.

\begin{figure}
\includegraphics[scale=1]{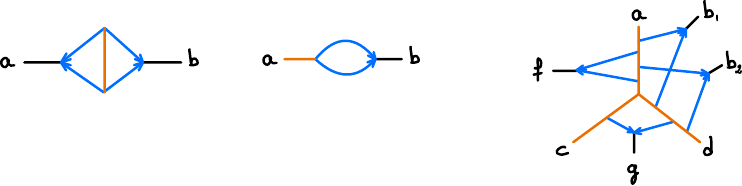}
\caption{Examples of galled bloblets $N$,
with root skeleton trees $T$ shown in orange.
Left: $N$ is neither strongly nor weakly tree-child, with no taxa on $T$.
Middle: $N$ is weakly but not strongly tree-child,
as some rooted partners are tree-child
(root at the hybrid's parent node or along a hybrid edge)
and some are not (root  along the edge incident to $a$).
$T$ has one taxon, $a$. Right: $N$ is tree-child, but 
not $\mathfrak{C}_4$.  $T$  has taxa $\{a,c,d\}$.   
The reduced graph of the subnetwork $N_{\{a,c,d,b_1,b_2,g\}}$ ($f$ omitted) is $\mathfrak{C}_4$.
}\label{fig:example1galled}
\end{figure}

Suppose  further that a galled bloblet $N$ is weakly tree-child, and
consider a tree-child rooted partner $N^+$ of $N$.
Since $N^+$ is tree-child, there is at least one path of tree edges from
its root to some leaf, and the taxon set $X_0$ of the
root skeleton tree is not empty. See, for example, \Cref{fig:example1galled}, middle.
The next lemma shows that if  $N$ is
(strongly) tree-child, then the root skeleton tree has all leaves labelled, and thus has at least 2 taxa, as in
\Cref{fig:example1galled} right, where $X_0=\{a,c,d\}$.

\begin{lemma}\label{lem:1galled-skeleton}
Let $N$ be a  galled, tree-child bloblet, and $X_h$ its taxa with hybrid parent nodes.
Then the unreduced root skeleton tree of $N$
is a phylogenetic tree that is composed of all of $N$'s tree nodes and
tree edges, except for $X_h$ and pendant edges leading to taxa in $X_h$.
\end{lemma}

\begin{proof}
Let $\widetilde{T}$ denote the unreduced root skeleton tree of $N$.
If $N$ has $k$ hybrid nodes, by Lemma \ref{lem:skeletonnum} it has
$k$ non-root skeleton trees in addition to the root skeleton tree.
Because $N$ is a galled bloblet, each non-root skeleton tree
consists of a lowest hybrid node in the blob, its descendant leaves,
and the pendant edges connecting them.
Thus $\widetilde{T}$ is composed of all of $N$'s other tree edges and tree nodes.

To finish the proof, we need to show that $\widetilde T$ is a phylogenetic tree, that is,
it does not have unlabeled leaves.
By \cite{2024MaxfieldXuAne}, $N$ can be rooted at any node in $\widetilde T$.
Let $e=uv$ be a pendant edge in $\widetilde T$,
with leaf node $v$, and consider the rooted partner $N^+_u$ rooted at $u$.
Since $N_u^+$ is tree-child, there is a (possibly empty)
directed path of tree edges starting from $v$ to some labeled leaf $y$.  
Since this path contains only tree edges, it must be in $\widetilde T$
which implies that $v=y$ is labeled.
\end{proof}

As a corollary to Lemma~\ref{lem:1galled-skeleton},
each hybrid edge $e$ in a galled tree-child bloblet $N$ 
has its parent node $u$ in $N$'s unreduced root skeleton tree.
If $u$ is binary in that tree, then $u$ is suppressed when reducing $\widetilde T$
and we say that $e$ \emph{attaches to an edge} in the root skeleton tree $T$.
If $u$ is not suppressed in $T$,
then we say that $e$ \emph{attaches to a node} in $T$,
or that it attaches to multiple edges in $T$ (those incident to $u$ in $T$).

\begin{defn}
Let $N$ be a semidirected network with a blob $B$.
The \emph{bloblet generated by} $B$ 
is the subgraph of $N$ comprised of the edges in $B$ and the cut-edges incident to $B$,
with new and distinct labels assigned to all unlabeled leaves. 
\end{defn}

\noindent
Note that the bloblet generated by a blob $B$ is semidirected,
and its topology depends only on $B$ and the number of cut-edges in $N$
incident to $B$'s boundary nodes.  In \Cref{fig:example01}, for example,
the network $N_3$ is the bloblet generated by $B_3$ in $N$.
If $N$ is itself a bloblet with internal blob $B$, then
the bloblet generated by $B$ is simply $N$. 

\medskip

We now define the new class of networks central to this work.

\begin{defn}[\textbf{Network Class $\mathfrak{C}_k$}]\label{def:classCk}
  Let $N$ be a semidirected network, $B$ a blob of $N$,
  and $N(B)$ the bloblet generated by $B$.
  Suppose that the leaf set of $N(B)$ is $Y_0\sqcup Y_\ell$ where $Y_\ell$
  are the leaves below the lowest nodes of $B$.
  We say that $B$ \emph{is} $\mathfrak{C}_k$, or \emph{in the class} $\mathfrak{C}_k$, $k\ge 3$,
  if $N(B)$ is reduced, galled, tree-child, with all hybrid nodes of out-degree~1,
  and for every $y\in Y_\ell$, the internal blob of the reduced graph of $N(B)_{Y_0 \sqcup \{y\}}$
 is a tree cycle of size $k$ or more.
   A network \emph{is} $\mathfrak{C}_k$ if it is reduced
  and all its blobs are $\mathfrak{C}_k$.
\end{defn}

\medskip

To illustrate these ideas, we consider again the networks displayed in \Cref{fig:example01}.
The network $N_3$  is a galled tree-child bloblet with $Y_\ell=\{g_1,g_2\}$, and in class $\mathfrak{C}_k$ for $k = 3, 4, 5$.
Since the bloblet generated by $B_2 \subset N'$ is not tree-child, this bloblet (and therefore $N'$) 
can not be in $\mathfrak{C}_k$ for any $k$.
In like manner, since $N_1$ is only weakly tree-child, $N$ is not $\mathfrak{C}_k$, for any $k$.

Trivial blobs (single nodes) have no hybrids so they are $\mathfrak{C}_k$ for all $k$.
Using the notation of \Cref{def:classCk}, note that if the hybrid parent
of $y\in Y_\ell$ were not bicombining,
then the internal blob of $N(B)_{Y_0 \sqcup \{y\}}$ would have at least 2 cycles
and this bloblet would not be $\mathfrak{C}_k$.  We state this formally.

\begin{lemma}\label{lem:CkHybridsBicombining}
	If $N$ is a $\mathfrak{C}_k$ network, then all of its
	hybrid nodes are bicombining, hence binary.
\end{lemma}

In a galled network, if a hybrid creates a tree cycle of size $k$ or more, 
then the cycle contains $k-2$ or more tree edges.
When $k=4$, the criterion of \Cref{def:classCk} thus simply means that  partner hybrid edges
do not attach to the same edge of the network's (reduced) skeleton forest.
Similar reasoning for $k=3,5$ yields the following. 

\begin{lemma}\label{lem:C41galled}
A bloblet $N$ is $\mathfrak{C}_3$ if it is reduced, galled,  tree-child,
with all hybrid nodes binary,
and partner hybrid edges do not attach to the same node of
$N$'s unreduced root skeleton tree.
It is $\mathfrak{C}_4$ if in addition the
partner hybrid edges do not attach to the same edge in its root skeleton tree,
and $\mathfrak{C}_5$ if, in addition,
partner hybrid edges do not attach to adjacent edges in its root skeleton tree.
\end{lemma}

For example, the bloblet $N$ in \Cref{fig:example1galled} (right) is 
$\mathfrak C_3$, but not $\mathfrak{C}_4$ because 
the hybrid edges above $f$   
attach to the same skeleton tree edge, creating a tree cycle of length 3. 
The reduced subnetwork of $N$ on all taxa except $f$ is
$\mathfrak{C}_4$, but not $\mathfrak{C}_5$, since each subnetwork on $\{a,c,d\}$ and any one of 
the hybrid taxa $b_1$, $b_2$, or $g$ is a  level-1 network with a 4-cycle.
More generally, $\mathfrak{C}_{k+1} $ is a proper subclass of $\mathfrak{C}_k$.

Finally, in closing this section, we prove that the $\mathfrak{C}_k$ property
implies a lower bound on the number of taxa.

\begin{lemma}\label{lem:Cknet}
A network $N$ in $\mathfrak{C}_k$ has no non-trivial $m$-blob with $m<k$.
If $N$ has at least one non-trivial blob, then $N$ has at least $k$ taxa.
\end{lemma}
\begin{proof} 
If $N$ is a bloblet in $\mathfrak{C}_k$ with a non-trivial $m$-blob $B$,
then each of its hybrid nodes is in a tree cycle with at least $k$ edges,
so the root skeleton tree has at least $k-1$ taxa. Since at least  one taxon descends from each hybrid node,
the bloblet's number of taxa satisfies $m\geq k$. 
If $B$ is an $m$-blob in  a general network $N$ in $\mathfrak{C}_k$, then
the bloblet generated by $B$, is in $\mathfrak{C}_k$, and the result follows from the bloblet case.
The final statement is immediate.
\end{proof}

\section{Identifiability of galled tree-child networks with large cycles}\label{sec:id}

Depending on data type and model, certain specific properties can be established that aid in showing
network identifiability.  We now explicitly state a number of these as assumptions, so that we
may show their role in proving network identifiability through combinatorial arguments. 
Later, in  \Cref{sec:moddat}, we prove (subsets of) these assumptions hold for
specific models with quartet concordance factors as data.
All networks are semidirected, without further restriction unless stated explicitly.

\begin{description}
\item[A-ToB.] If $N$ is a semidirected network, the topology of
its tree of blobs is identifiable.

\item[A-4circ.]
 If $N$ is known to be  a 4-taxon level-1 network,
 the set of circular orders congruent with this level-1 network is identifiable.
 Specifically, if $N$ has a non-trivial split, this split is identifiable;
 and if $N$ has a 4-cycle, then the circular order of taxa around this cycle is
 identifiable.  
 
\item[A-3blob.] If $N$ is an extended bloblet
with an internal 3-blob $B$, then whether $B$ is trivial or non-trivial
is identifiable.

\item[A-4len.] If $N$ is known to be a 4-taxon extended bloblet with an internal 4-cycle
whose hybrid node is known, then the
 lengths of tree edges in the cycle are identifiable.
 If $N$ is known to be a 4-taxon tree, possibly extended by 2-blobs on its pendant edges,
 then the length of the internal tree edge is identifiable.
\item[A-hyb$(\mathfrak{C})$.]
 If $N$ is known to be an extended bloblet
 with internal blob $B$ in some class $\mathfrak{C}$ of networks,
 then the set of taxa below the hybrid nodes in $B$
 is identifiable. 
\end{description}

\begin{rmk}
For a class $\mathfrak{C}$ that includes level-1 networks,
A-4circ and A-hyb$(\mathfrak{C})$ together give full topological
identifiability of 4-cycles in level-1 networks in $\mathfrak{C}$.
\end{rmk}

We first prove a result on identifying hybrid nodes.

\begin{lemma}\label{lem:identifyXh}
  Let $\mathfrak{C}_5^e$ be the class of
  extended bloblets whose internal blob $B$ is $\mathfrak{C}_5$.
  If A-ToB holds, then A-hyb($\mathfrak{C}_5^e$) holds.
  More specifically, taxon $x$ is  below a hybrid node of $B$
  if and only if there is a subset $Y$ of 4 taxa
  such that the tree of blobs $T(N_{Y\cup\{x\}})$ is a star tree but the tree of blobs $T(N_{Y})$ is resolved.
\end{lemma}
\begin{proof} 
 Let $N$ be a network in $\mathfrak{C}_5^e$ with internal blob $B$,
 and taxon set $X=X_0\sqcup X_h$
 where $X_h$ is the subset of taxa with a hybrid ancestor in $B$.
 
Let $x\in X_h$. By Definition~\ref{def:classCk}, the
reduced induced
subnetwork $\overline{N_{X_0\cup\{x\}}}$
is a level-1 network with a single $k$-cycle with $k\geq5$.
Therefore there exists a
subset of 4 taxa $Y \subseteq X_0$ such that
the reduced graph $\overline{N_{Y\cup\{x\}}}$ is level-1 with a single $5$-cycle,
and $T(N_{Y\cup\{x\}})$ is a star tree.
 Since all taxa in $Y$ are on the root skeleton tree of $N$
and $\overline{N_{Y\cup\{x\}}}$ has a $5$-cycle, 
$T(N_{Y})$ is resolved.

Conversely let $Y$ be a 4-taxon set such that $T(N_{Y\cup\{x\}})$ is a star tree
but $T(N_{Y})$ is resolved.
Note that $N_{Y\cup\{x\}}$ must be an extended bloblet
with a non-trivial internal 5-blob $B$ for this to occur.
By Lemma \ref{lemma:Linkingnodes}, $x$ is below a linking lowest node of $B$,
which is hybrid, and $x\in X_h$.
\end{proof}

Another useful result implying identifiability of hybrids  is the following.

\begin{lemma}\label{lem:tobplus3}
Consider the class $\frak C_4^e$ of extended bloblets whose internal blob is $\frak C_4$.
If A-ToB and A-3blob hold, then so does A-hyb($\frak C_4^e$).
\end{lemma}

\begin{proof} 
Let $N$ be a $\frak C_4$ extended bloblet, with internal blob $B$.
We identify the descendants $X_h$ of hybrid nodes in $B$ by
identifying the complementary set of taxa $X_0 = X \smallsetminus X_h$.

Consider all subsets $Y$ of taxa on $N$ with $\vert Y \vert \ge 3$,
and the tree of blobs 
$T(N_Y)$
(identifiable by A-ToB) of the induced network $N_Y$.
For each internal node in each $T(N_Y)$, use A-3blob on all subtrees 
determined by picking 3 edges emanating from the node to determine and discard those subsets $Y \subseteq X$ 
with $N_Y$ containing a non-trivial $k$-blob, $k\ge 3$,
as such $Y$ must contain a taxon in $X_h$.

Some remaining subsets $Y$ may still contain a taxon $a$ in $X_h$ because the cycle
in $N$ formed by the hybrid edges above $a$ in $B$ and edges in an unreduced skeleton tree of $N$
has been collapsed to a degree-2 node, which was then suppressed in $T(N_Y)$.
However, for such sets arising from $\frak C_4$ networks we can remove $a$ and include at least 2 additional taxa
from $X_0$ which are in distinct
groups off of the cycle. This produces a larger set $Y'$ which has an additional degree-3 node in its tree of blobs, 
which  A-3blob identifies as trivial. Moreover, $Y'$ has one less taxon from $X_h$.

Thus taking the largest set $Y$ which produces a tree of blobs,
all of whose internal nodes arise from trivial blobs as tested by
A-3blob, gives precisely $X_0$.
\end{proof}

\medskip

We now prove our main combinatorial results on bloblet identifiability.

\begin{theorem}\label{thm:1galled-C4}
Consider the class $\mathfrak{C}_4^e$ of
extended bloblets whose internal blob is $\mathfrak{C}_4$,
and suppose A-ToB, A-4circ, A-4len, and  
A-hyb$(\mathfrak{C}_4^e)$ hold.
For a network $N^e$ in $\mathfrak{C}_4^e$, let $N$ be the reduced bloblet
that $N^e$ extends. 
Then the semidirected topology of $N$ and the
length of its internal tree edges are identifiable.
\end{theorem}

\noindent
Combining with Lemma~\ref{lem:identifyXh} and
using $\mathfrak{C}_5 \subseteq \mathfrak{C}_4$, this implies  the following.

\begin{cor}\label{cor:1galled-C5}
  Consider the class $\mathfrak{C}_5^e$ of extended bloblets whose internal blob
  is $\mathfrak{C}_5$, and suppose A-ToB, A-4circ and A-4len hold.
  For a network $N^e$ in $\mathfrak{C}_5^e$, let $N$ be the bloblet that
  $N^e$ extends. Then the semidirected topology of $N$ and the
  length of its internal tree edges are identifiable.
\end{cor}

\begin{proof}[Proof of Theorem~\ref{thm:1galled-C4}]
For $N^e$ and $N$ as in the statement, let $B$ be
their common internal blob.
By assumption A-hyb($\mathfrak{C}_4^e$)
we can identify the partition $X=X_0\sqcup X_h$
where $X_h$ is the set of taxa that are below hybrid nodes of $B$.
$X_h$ may be empty, in which case $B$ is trivial and there is nothing to prove.

Otherwise, the taxa $X_0$ are those on the root skeleton tree $T_\rho$
of $N$, so by assumption A-ToB we can identify the topology of $(N^e)_{\!X_0}$'s
tree of blobs, which is $T_\rho=\overline{N_{X_0}}$.

Now take $x\in X_h$. As $N$ is  $\mathfrak{C}_4$, 
the reduced graph $\overline{N_{X_0\cup\{x\}}}$ is level-1, with a single cycle, of at least 4 edges.
By considering induced 4-taxon networks on $x,a,b,c$ for all choices of 3 taxa $a,b,c\in X_0$ and using A-4circ, 
we may determine those taxa which attach to the single cycle of
$\overline{N_{X_0\cup\{x\}}}$
by paths to each of the nodes in the cycle. 
This is enough to determine the edges (and possibly nodes)
of $T_\rho$ onto which the hybrid edges above $x$ attach.

From this information across all taxa in $X_h$, we can identify which edges in
$T_\rho$ arise from a path of multiple edges in the unreduced $N_{X_0}$,
and which edges in $T_\rho$ match a single edge in $N_{X_0}$.

\smallskip
Next we need to identify for each taxon $x$ in $X_h$,
the precise locations at which the two hybrid edges above $x$ originate on $T_\rho$,
and the length of those internal tree edges in $B$
which are internal tree edges of ${N}_{X_0}$. 

Let $uv$ be an edge in $T_\rho$. If $uv$ arises from a single internal edge in $N_{X_0}$
(no hybrid edge attaches to $uv$ except possibly at its ends), then $uv$
arises from an internal edge in $(N^e)_{\!X_0}$ and its edge length
is identifiable by A-4len.

If one or more hybrid edges attach to $uv$, 
let $p$ be the corresponding path in $N$ of tree edges and tree nodes
$u,a_1\ldots,a_{k},v$ ($k\geq 1$) that reduces to $uv$
in $T_\rho$. It is possible for $u$ or $v$ to be a leaf, but not both
(because $T_\rho$ has 3 or more taxa by \Cref{lem:Cknet}).
We thus  assume that $u$ is internal.
At this point, we have already identified the set of taxa $X_{uv} \subseteq X_h$
below a hybrid edge that attaches to $uv$. 
For $x\in X_{uv}$, let $a(x)$ be the associated attachment node.
It suffices to show that we can identify the distance between $u$ and $a(x)$
for each $x\in X_{uv}$.

To this end, fix $x\in X_{uv}$. In $\overline{N_{X_0\cup\{x\}}}$,
$u$ and $a(x)$ are internal nodes
(of degree $\geq 3$), and $v$ is either a leaf or internal.
The cycle in $\overline{N_{X_0\cup\{x\}}}$ contains 4 or more nodes, two of which are $a(x)$
and the hybrid node above $x$.

If $v$ is a leaf the cycle also includes $u$ since $N$  is in $\mathfrak{C}_4$. By A-4len, we can 
identify the length of tree edges in this cycle on $N^e$ which includes the length of $u a(x)$.
If $v$ is internal, we can similarly identify   the length of either $u a(x)$ or of $v a(x)$, whichever is 
part of the cycle. But since $uv$ is an internal edge of the tree $\overline{N_{X_0}}$,
using A-4len we can identify its length on $N^e$. Subtracting the length of $v a(x)$
from it gives the length of $u a(x)$.

At  last, note that since all hybrid nodes are binary,
 any taxa  $x$ and $x'$ sharing the same two attachment points $a_1,a_2$
 must have distinct pairs of hybrid edges above them:
$(a_1,h)$ and $(a_2,h)$ above $x$; 
$(a_1,h')$ and $(a_2,h')$ above $x'$, with $h\neq h'$, as opposed to  
$x$ and $x'$   descending from the same pair of hybrid edges
$(a_1,h)$ and $(a_2,h).$ 
\end{proof}

\begin{rmk}\label{rmk:relative}
In the proof of Theorem \ref{thm:1galled-C4}, edge lengths, identified by A-4len, 
are used to precisely locate the attachment points of
hybrid nodes on the root skeleton tree. In fact
identifying the \emph{relative} edge lengths is
sufficient to identify the semidirected topology of the blob. That is,  if 
A-4len is weakened  to identifying only whether whether $\ell(e)<\ell(e')$, $\ell(e)=\ell(e')$ or $\ell(e)>\ell(e')$
when $e$ and $e'$ are composite edges formed from a path and subpath of tree edges in the network, then a
weakened theorem can be established.

In particular, edge lengths may be considered
in any units (consistently across 4-taxon sets),
such as years, number of generations, coalescent units, or substitutions per sites.
If only an ordering by magnitude of edge lengths can be identified,
then edge lengths in the full network are not identifiable, but the claim of
topological identifiability stated in Theorem \ref{thm:main} remains valid.
This suggests robustness of network topology estimation to edge length inference error.
\end{rmk}

\medskip

We extend the two previous results from bloblets to general networks.

\begin{theorem}\label{thm:galled-C4}
Suppose A-ToB, A-4circ, A-4len, and A-hyb($\mathfrak{C}_4$) hold.
For a semidirected network $N$ in $\mathfrak{C}_4$, the topology of $N$ and the lengths of internal tree edges
in its blobs are identifiable.
\end{theorem}

\begin{proof}
By A-ToB, the tree of blobs of $N$ is identifiable. To identify the structure of an individual blob $B$, we pass to
an extended bloblet as follows:
For each cut edge of $N$ incident to $B$ choose one taxon which is separated from $B$ by that edge, 
forming a taxon subset $Y$ Then consider $\overline{N_Y}$, an extended bloblet with blob $B$.
By Theorem~\ref{thm:1galled-C4} the topology of $B$ and its internal tree edge lengths are identified.
This gives the full semidirected network and the lengths of all tree edges inside  blobs, as claimed.
\end{proof}

A similar argument, using the weaker hypotheses of \Cref{cor:1galled-C5}, 
yields the following.

\begin{theorem}\label{thm:galled-C5}
  Suppose A-ToB,  A-4circ, and A-4len hold.
  For a semidirected network in $\mathfrak{C}_5$,
  the topology 
  and the lengths of internal tree edges
  in non-trivial blobs are identifiable.
\end{theorem}
 
Note that the theorems do not claim identifiability of the lengths of
cut edges of $N$.   However, A-4len can be applied to any cut edge $uv$ whose endpoints are not incident to a blob, 
by choosing four taxa that define $uv$. 

\begin{rmk}\label{rmk:edgeNonID}
Identifying lengths of hybrid edges and tree edges incident to non-trivial blobs
seems to depend on more than the general assumptions made above. 
Even in the level-1 case studied in \cite{2024ABGarrotelopesR}, lengths of edges near 3-cycles 
and leaves can be nonidentifiable.
\end{rmk}

\begin{rmk}
Arguing as in the proof of \Cref{thm:galled-C4} or \ref{thm:galled-C5}, the topology 
of $\mathfrak{C}_4$ or $\mathfrak{C}_5$ blobs of a general network $N$ may be identifiable,
even if the full network $N$ is not of those classes.
The topologies of its $\mathfrak{C}_4$ blobs are identifiable
if A-ToB, A-4circ, A-4len, and A-hyb($\mathfrak{C}_4$) hold, and topologies of
$\mathfrak{C}_5$ blobs if only the first three assumptions hold.
The network $N_0$ of \Cref{fig:intro}, discussed in the introduction, 
provides an example of this since only one blob is $\frak C_5$.
\end{rmk}

\begin{rmk}\label{rmk:SimplerC5}
Without assuming A-4len, the proof of \Cref{thm:1galled-C4}
shows that, for each hybrid in a blob, we can identify the 2 edges
onto which its parent edges attach. The proof
could thus be modified to show topological identifiability of $N$ for a smaller class
of networks where no two hybrid edges attach to the same edge of the root skeleton tree.
For example, the rightmost blob of $N_0$ in \Cref{fig:intro} is $\frak C_5$ and in this
more restrictive class.

Even if multiple hybrid edges with different children attach to the same edge
in the root skeleton tree,
without A-4len we can still determine a finite list of networks, one of which
is the true network.
\end{rmk}

\section{Models and quartet concordance factor data}
\label{sec:moddat}

We now present several models  of gene tree formation on a network, as well as the quartet concordance factor data type.
Then we establish that the assumptions needed to apply
Theorems~\ref{thm:galled-C4} and~\ref{thm:galled-C5} hold,
and conclude with our main result in \Cref{thm:main}.

\subsection{Models for gene trees}

We describe three models of gene trees forming within species networks.
\begin{defn}
 \label{def:gene-tree-models}
 Let $N^+$ be a rooted metric phylogenetic network with edge lengths
 $\ell(e)$ in coalescent units (generations/population size). 
 Then the following models,
 determine distributions of topological or metric unrooted gene trees.
 	\begin{enumerate}
		\item \emph{Displayed tree} (DT) model:
		For each edge $e$ of $N^+$, we also specify
		the effective population size $\popsize(e)>0$ and
		mutation rate $\mu(e)>0$ in substitutions per site per generation.

		Only gene trees whose topology is displayed in $N^+$ have positive probability,
		equal to the product of inheritance probabilities of all edges in $N^+$ forming $T$:
		$$\P{T} = \gamma(T) = \prod_{e\in T}\gamma(e).$$
		Each edge $e$ of $T$  is assigned
		length $s(e)=\ell(e)\eta(e)\mu(e)$, giving a distribution of metric unrooted gene trees.
		
		\smallskip
		
		\item \emph{Network multispecies coalescent model with
			independent inheritance} (NMSCind): Gene trees form according to the coalescent
		model within each population (edge) of the network. At a hybrid node with parental edges
		$e_1,\dots,e_m$, 
		each lineage is inherited from population $e_k$ ($1\leq k\leq m$) with
		probability $\gamma(e_k)$, 
		independently of the other lineages \cite{Luay2012}.
		This gives a distribution of topological unrooted gene trees.
		
		\smallskip
		
		\item \emph{Network multispecies coalescent model with
			common inheritance} (NMSCcom): Gene trees form according to the coalescent
		model within each population. At a hybrid node with parental edges
		$e_1,\dots,e_m$, 
		all lineages of a given gene are inherited from the same population $e_k$
		($1\leq k\leq m$), chosen with probability $\gamma(e_k)$. Equivalently, a 
		displayed tree is chosen with probability as in item 1, and
		a gene tree forms within it according to the coalescent process as in item 2 restricted to a tree
		\cite{2011GerardGibbsKubatko}.
		This gives a distribution of topological unrooted gene trees. 
	\end{enumerate}
\end{defn}

The NMSCind and NMSCcom models (items 2 and 3) are the two extreme cases of a model
with correlated inheritance of lineages at reticulations
\cite{2023fogg_phylocoalsimulations}, which for simplicity we do not consider in full
generality here.

We need not explicitly consider models of sequence evolution on gene trees, as we treat the gene trees themselves as data. 
In practice, one of course needs to assume these gene trees can be robustly inferred from sequences, as all ``2-stage'' inference 
methods utilizing inferred gene trees must do.

\subsection{Quartet concordance factors}

A number of network identifiability results and practical network inference methods are based on quartet \emph{concordance factors (CFs)}
\cite{SolisLemus2016, Banos2019, ABR2019, ABMR2023, 2024ABGarrotelopesR, ABRW2025}.
Although many of these results are limited to level-1 networks,
we draw on them to obtain results for the more general classes of networks in this work.

For a network with $n$ taxa, quartet CFs are the $3 \cdot \binom{n}{4} $ probabilities of the
unrooted gene quartet topologies that might relate a subset of four taxa.
These probabilities can be calculated under any model $M$ of gene tree formation on a network $N$.
If $M$ generates metric gene trees on the full set of $n$ taxa, 
CFs are obtained by pruning taxa except those in a quartet, and then
marginalizing over edge lengths and root location.
For the NMSCind and NMSCcom models 
only resolved quartet trees have positive probability, so CFs for each 4-taxon set have the form
\begin{equation*}
	\cf_{abcd} = (\cf_{ab|cd}, \cf_{ac|bd}, \cf_{ad|bc}) = 
	(\mathbb{P}_M({ab|cd}), \mathbb{P}_M({ac|bd}), \mathbb{P}_M({ad|bc})).
\end{equation*}
For the DT model on a binary network we also have probability 0 of an unresolved quartet topology. 

\subsection{Identifiability assumptions}
\label{sec:moddat-idAs}

We next investigate the validity of the assumptions laid out in Section \ref{sec:id} in the context of  
the three models and quartet CFs just introduced.  This requires additional
assumptions --- including that networks are binary and
numerical parameters are generic (lie outside some subset of measure zero) ---
to utilize already published identifiability results.
These restrictions do not alter the arguments of \Cref{sec:id}, and the main topological
identifiability results there still apply.

\subsubsection{A-ToB}

The identifiability of the tree of blobs of a binary network from quartet CFs
under the NMSCind model was proved in \cite{ABMR2023},
with \cite{AllmanEtAl2024} providing a practical inference algorithm.
\cite[Corollary 6.6]{2025Rhodes-circular} extended that work to the DT and NMSCcom models,
but made additional assumptions of no ``anomalous quartets'' on the network
in order to study circular orders of blobs.

In \cite{ABMR2023}, the proof of tree of blobs identifiability
is primarily combinatorial, with the key exception the fundamental case
for 4-taxon networks \cite[Theorem 1]{ABMR2023}.
Noting that the proofs of those results did not require
assuming no anomalous quartets, straightforward
modifications extend those arguments to the DT and NMSCcom models.
We state this formally as the next proposition.

\begin{prop}\label{prop:AToB}
Under the DT, NMSCind, and NMSCcom models with quartet CFs from a binary network
with generic numerical parameters, the network's tree of blobs is identifiable,
so assumption A-ToB holds.
\end{prop}

\subsection{A-4circ}

The identifiability of circular orders for outer-labelled planar networks
with blobs of any size was studied in \cite{2025Rhodes-circular},
but the A-4circ assumption only concerns level-1 4-taxon networks,
where the result follows immediately from \cite[Theorem 4]{Banos2019}
for the NMSCind model.
For the other models on a level-1 4-taxon network $N$ with generic parameters,
one can directly compute that if $N$ displays the split $ab|cd$,
then $\cf_{abcd}$ has the form $(1,0,0)$ under DT,
and $(p,q,q)$ under NMSCcom and NMSCind,
while if $N$ has a 4-cycle with circular order $(a,b,c,d)$ the
CFs have the form $(p,0,q)$, $p,q>0$ for DT and $(p,r,q)$ with
$p,q>r>0$, $p\ne q$ under NMSCcom and NMSCind. 
If the 4-taxon network has a 4-polytomy the CF is
$(1/3,1/3,1/3)$ for the NMSC models, while for DT  gene trees are unresolved with probability 1. 
Thus we obtain the following.

\begin{prop}\label{prop:A4circ}
Under the DT, NMSCcom, and NMSCind models with generic parameters using
quartet CFs, assumption A-4circ holds.
\end{prop}

\subsection{A-4len}

Let $N$ be a 4-taxon level-1 network with a 4-cycle with circular order $a,b,c,d$
whose hybrid node is ancestral to taxon $a$.
First consider the DT model.
As CFs capture only topological information on gene trees,
$\cf_{abcd}=(p,0,q)$ is independent of all branch lengths.
Therefore A-4len does \emph{not} hold from CFs for this model.
However, from one metric unrooted gene tree of each topology,
one can identify the lengths (in substitutions per site)
of the tree edges in the cycle, as these are simply internal branch lengths
on one of the unrooted trees displayed by $N$.
Similarly if $N$ is a tree, possibly extended with 2-cycles on its pendant edges,
the CFs alone are not enough to determine the internal edge length under DT,
but one metric unrooted gene tree is.
Note that under the DT model, A-4len is then satisfied using edge lengths
in substitutions per site: the units that can be identified on gene trees
from (arbitrary length) sequence data.

For the NMSCind model, the identifiability of the tree edge lengths in a 4-cycle
from CFs is dependent on having multiple samples
(either from different taxa or within the same taxon) from specific cut edges
attached to the 4-cycle, as characterized by \cite[Proposition 29]{2024ABGarrotelopesR}.
Specifically, under NMSCind, if the circular order of taxon groups around the cycle is $(A,B,C,D)$,
where each group is the subset of taxa separated from the blob by a common boundary node,
and if $A$ is below the hybrid node,
one needs either two samples from $B$ or two from $D$, 
or 2 from both $A$ and $C$.  For NMSCcom, two samples from $A$ is sufficient.
(See Appendix \ref{sec:appendix} for details.)
Rather than use these facts with such exactness, we simply say that these lengths
are identifiable if we have 2 samples per taxon group.

\begin{prop}\label{prop:A4len}
Using quartet CFs, assumption A-4len holds
\begin{enumerate}
\item under the DT model if one has a metric gene tree of each possible topology from the 4-taxon network, and 
\item under the NMSCcom and NMSCind models if there are 2-samples per taxon.
\end{enumerate}
\end{prop}

For the strongest statements on identifiability of tree edge lengths
in 4-cycles from CFs under the NMSCind and NMSCcom models,
see \cite{2024ABGarrotelopesR} and Appendix~\ref{sec:appendix} of this work.

\subsection{A-3blob}

We establish the validity of assumption A-3blob for the two models with a coalescent process.

\begin{prop}\label{prop:A3blob}
Under the NMSCind or NMSCcom model, with at least
2 samples per taxon and quartet CF data, assumption
A-3blob holds for binary semidirected networks with generic parameters.
\end{prop}

\begin{proof} Let $a,b,c$ denote the 3 taxa on the network,
and denote the 2 samples per taxon by $a_1$, $a_2$ from $a$ etc.

Let $G_{abc}$, $G_{bca}$, and $G_{cab}$ be the polynomials in quartet CFs
given by
\[
G_{abc} =
     \cf_{a_1c_1|a_2c_2}\cf_{a_1b_1|b_2c_2}
 - 2 \cf_{b_1c_1|b_2c_2}\cf_{a_1b_1|a_2c_2}
 +   \cf_{a_1b_1|a_2b_2}\cf_{a_1c_1|b_2c_2}
\]
and similarly for $G_{bca}$ and $G_{cab}$ by permuting taxon labels.
Then for level-1 networks and the NMSCind and NMSCcom models 
\cite[Proposition 9 and Remark 1]{2024ABGarrotelopesR}
established that
 $N$ has a non-trivial 3-blob if and only if at least one of the
$G$ polynomials does not vanish:
$G_{abc}\neq 0$ or $G_{bca}\neq 0$ or $G_{cab}\neq 0$.

If the network has an arbitrary non-trivial 3-blob, by specializing some of the
hybridization parameters to 0 or 1 all 2-blobs can be effectively replaced by
edges  and the 3-blob can be effectively reduced to a 3-cycle by
Lemma~\ref{lem:3cycin}. This implies that at least one $G_{xyz}$,
viewed as a function of the network parameters, does not vanish for one
specialization of the parameters. As the CFs are analytic functions of the
parameters, so is $G_{xyz}$, and the non-vanishing of an analytic function at a
single point implies its non-vanishing at generic points.
Thus for generic parameter values, this $G_{xyz}$ is non-zero.
\end{proof}

The argument in this proof fails for the DT model, since, as pointed out in \cite{2024ABGarrotelopesR}, 
in the level-1 case the CFs appearing in the $G$ polynomials are all for quartets 
not displayed on the network, and thus are 0 under that model.
It remains an open question whether  other approaches might imply  A-3blob holds under DT.

We note that the 2-sample assumption is necessary for the leaves of a binary 3-bloblet,
but can be weakened for larger networks, as in the similar discussion for A-4len above.
Concretely, for a 3-blob, each boundary node corresponds to a taxon group
and the argument for \Cref{prop:A3blob} is valid when each taxon group has at least 2 samples.
In \Cref{fig:intro}, for example, the middle 3-blob in $N_0$ can be identified
as non-trivial even with a single sample per taxon.

\subsection{A-hyb($\frak C$)}

We consider the question of identifying descendants of hybrid nodes
of a blob for classes $\frak C_4$ and $\frak C_5$.
Although A-hyb($\mathfrak{C}_4)$ implies A-hyb($\mathfrak{C}_5$) since
$\mathfrak{C}_5 \subsetneq \mathfrak{C}_4$, we first state
a general identifiability result for $\frak C_5$, since it requires few restrictions.
For class $\frak C_4$, we give a different argument, valid under
the coalescent  models but requiring multiple samples per taxon.

\begin{prop}\label{prop:AhybC5}
If $N$ is a binary 2-blob extension of a bloblet $B$
in $\frak C_5$ with generic numerical parameters, then under the models
DT, NMSCcom, and NMSCind the set of taxa descended from a hybrid node in $B$
is identifiable from quartet CFs.
Thus A-hyb($\frak C_5$) holds.
\end{prop}

\begin{proof}
This follows directly from \Cref{prop:AToB} and Lemma \ref{lem:identifyXh}.
\end{proof}

\begin{prop}\label{prop:AhybC4}
If $N$ is a binary 2-blob extension of a bloblet $B$
in $\frak C_4$ with generic numerical parameters and 2 samples per taxon,
then under the models NMSCcom and NMSCind the set of taxa descended
from a hybrid node in $B$ is identifiable from quartet CFs.
Thus A-hyb($\frak C_4$) holds.
\end{prop}

\begin{proof}
For such a network and sampling, A-ToB and  A-3blob hold by Propositions \ref{prop:AToB} and  \ref{prop:A3blob}.
Lemma \ref{lem:tobplus3} then yields the claim.
\end{proof}

\subsection{Main results}

We now state and prove our main theorem.

\begin{theorem}\label{thm:main}
  Let $N$ be a binary semidirected phylogenetic network in $\frak C_4$.
  Then under the NMSCind and NMSCcom models with generic numerical parameters
  and 2 samples per taxon, the semidirected topology of $N$ and the lengths of internal tree edges
  in blobs are identifiable from quartet CFs.

  Under the DT  model with CFs and metric gene trees, or the coalescent models
  with a single sample per taxon, the same result holds for binary networks in $\frak C_5$.
\end{theorem}
\begin{proof} For the NMSCind and NMSCcom models with 2 samples per taxon,
A-ToB holds by Proposition \ref{prop:AToB},
A-4circ by Proposition \ref{prop:A4circ},
A-4len by Proposition \ref{prop:A4len}, and 
A-hyb($\mathfrak{C}_4$) by Proposition \ref{prop:AhybC4}.
Thus Theorem \ref{thm:galled-C4} yields the first claim.

Under the DT model, the same propositions show A-ToB and A-4circ hold, as does A-4len since 
we have  metric gene trees. While we do not have A-hyb($\mathfrak{C}_4$),  applying  
\Cref{thm:galled-C5} establishes the second claim. The coalescent models on $\frak C_5$ 
networks with 1 sample per taxon are handled similarly.
\end{proof}

For the NMSC models we obtain a weaker, 
yet still interesting result, that does not depend on  A-4len or A-3blob
(as motivated by Remark \ref{rmk:SimplerC5}).
We omit the proof, since it closely follows previous arguments.

\begin{prop}\label{cor:galled-C5simp}
Let $N$ be a binary semidirected phylogenetic network in $\mathfrak{C}_5$ 
with the additional requirement that no two hybrid edges
attach to the same edge in the skeleton forest.
Then, under the DT, NMSCcom, and NMSCind models with 1 sample per taxon and
generic numerical parameters, the topology of $N$ is identifiable from quartet CFs.
\end{prop}

While the network family in this proposition is a proper subset of $\mathfrak{C}_5$,  
it still includes networks of arbitrary level, and many that are not outer-labeled planar,
including for example a network containing the rightmost non-trivial blob of $N_0$
in \Cref{fig:intro}.

\section{Discussion}\label{sec:disc}

The classes of networks that we have shown to be identifiable from quartet concordance factors
are essentially those
that are binary, galled, with all blobs tree-child, and with no ``small'' cycles.
These include networks of arbitrary level, well-beyond the level-1 structure
currently assumed by most practical inference methods
\cite{SolisLemus2016, ABR2019,Kong2024,ABRW2025,Holtgrefe2025},
and ones without the outer-labelled planar embeddings that previously have been
shown to lead to additional identifiability results beyond level-1 \cite{2025Rhodes-circular}.
Since we build our work on first identifying a network's tree of blobs,
and then analyze each blob separately, it allows for the structure of some blobs
to be identified while others may be too complex to do so
(with current understanding).

Note that requiring that all blobs be tree-child is a stronger condition
than requiring the entire network to be tree-child.
For example, given a semidirected tree-child network,
placing a weakly tree-child blob below a hybrid node
(and consequently below the root component) would still result in a tree-child network.
Whether this condition can be relaxed, or to what extent it is necessary within
the class of networks considered,
remains a question for future exploration.

While the network restrictions considered here are not biologically motivated,
they are natural ones for identifiability results.
For instance galled networks are ones in which the various reticulations
in a blob have some independence of one another, with each determining
a unique cycle and none ancestral to each other.
While this is much weaker than requiring the non-intersecting cycles of level-1 networks,
it allows, to some degree, for an analysis one cycle at a time.
Tree-child networks are those in which every node is ancestral to a leaf by a
path with no reticulations.
This can be viewed as giving some `direct' data on all nodes in the network,
not obscured by the genetic interchange the reticulations model. 
Finally, while biologically one might expect small cycles,
representing gene flow between closely related species, to be most common,
it is also plausible that they will be the most difficult to infer correctly due
to the similarity of the intermixed genomes.
In particular 3-cycles represent gene flow between sister species,
and if this occurred soon after species divergence it might only be detectable
through more detailed data than that we consider here.

By largely focusing on quartet CFs, which are determined by
topological gene tree information alone, our results are likely to concern
what can be most robustly inferred.
Identifiability results from metric gene trees that are applicable to empirical data
require a detailed model of the substitution process along gene trees.
Variation in this process across the genome may be substantial
(both in across-site rate variation and in the substitution process itself).
Rate variation across lineages, violating a molecular clock,
is known to reduce network accuracy and increase the detection of spurious
reticulations for methods using metric gene trees that assume no rate variation
\cite{2017Ogilvie-starBEAST2-relaxed,2022Flouri-relaxed,2023Frankel,
2024Cao_misspecification_rateheterogeneity,2024Koppetsch}.
Thus while metric gene trees could provide more information on network structure,
how to extract that information 
accurately 
needs further development.

\medskip

One final aspect of our results worth highlighting, for both empiricists and developers of methods, is that multiple samples 
per taxon can provide more information on network structure than single samples do. While already inherent in other 
theoretical works, and especially prominent in \cite{2024ABGarrotelopesR}, many empirical data sets currently include only 
single samples. While the cost, in time, effort, and funds for routinely collecting multiple samples cannot be dismissed, 
doing so would increase what can be inferred. With a coalescent process part of the assumed model, one may view 
each sample as a `probe'  into the past, with multiple probes from the same source allowing their 
differing histories to provide more insight than does any one alone.

\section{Funding}
This work was supported in part by the National Science Foundation through grants
DMS-2051760 (ESA and JAR),
DMS-2331660  (HB),
DMS-2023239  (CA), and by
grant DMS-1929284 while all authors were in residence
at the Institute for Computational and Experimental Research in Mathematics in Providence, RI, during the 
``Theory, Methods, and Applications of Quantitative Phylogenomics" program.

\section*{Appendix}\label{sec:appendix}

Proposition \ref{prop:A4len}, part 2, gives sufficient conditions
for Assumption A-4len to hold for quartet CF data under NMSCcom.
However, assuming two samples from each of the four taxon groups is
stricter than is needed.
Assume that $N$ is a network with a 4-cycle, where the circular order of taxon groups around the cycle is $(A,B,C,D)$
with $A$  the hybrid group.  Assume further that the tree edge lengths in the cycle between groups $B$ and $C$ and between
$C$ and $D$ are $t_1$  $t_2$, respectively.

Note first that if only a single sample is available
from among the hybrid descendants then the CFs from models NMSCcom
and NMSCind are equal and the necessary condition follows from \cite{2024ABGarrotelopesR}.
Specifically, with a single hybrid sample a necessary and sufficient condition for identifiability of the  lengths
of the tree edges in the 4-cycle
is that at least two samples are taken from either group that is neighbor to the hybrid group.

In contrast, \cite{2024ABGarrotelopesR} shows that if 2 samples are taken from the hybrid group and only 1
from the other groups, under NMSCind  $t_1, t_2$ are not identifiable from quartet CFs.
However, under NMSCcom these branch lengths are identifiable. To establish this, we used the computational algebra software 
{\tt Singular 4.4.0 }\cite{Singular}  to show
that the edge probabilities $X_i = \exp(-t_i)$ for the  tree edges in the 4-cycle can be computed from CFs under NMSCcom.
For example, adopting a short notation with $C_{acab}=\cf_{a_1c|a_2b}$, etc.,
\[
X_1 = \frac{2 C_{acab} C_{acbd} + C_{acab} C_{adbc} - 2 C_{acbd} C_{adab} - C_{acbd} C_{adac} - C_{adab} C_{adbc} + C_{adac} C_{adbc}}
  {-C_{acbd} C_{adab} + C_{adab} C_{adbc} + C_{acab} - C_{adab}}.
\]
A formula for $X_2$ is obtained from that for $X_1$ by exchanging the taxon labels $b$ and $d$.

\bibliographystyle{alpha}
\bibliography{Hybridization.bib}

\end{document}